\documentclass[12pt]{iopart}
\usepackage{amsthm}
\usepackage{amssymb}
\usepackage{verbatim}
\usepackage{cite}
\usepackage{hyperref}
\usepackage{graphicx}
\usepackage{xcolor}
\newtheorem{corollary}{Corollary}

\newtheorem{lemma}{Lemma}

\newtheorem{theorem}{Theorem}

\newtheorem*{remark}{Remark}

\catcode`\*=11
\makeatletter
\let\oldequation*\equation*
\let\equation*\@undefined
\let\oldendequation*\endequation*
\let\endequation*\@undefined
\makeatother
\catcode`\*=12

\usepackage{physics}

\begin{document}
\title{Guaranteeing Completely Positive Quantum Evolution} 
\author{Daniel Dilley$^a$, Alvin Gonzales$^b$, Mark Byrd$^{a,b}$}
\address{$^a$ Department of Physics, Southern Illinois University, 1245 Lincoln Dr, Carbondale, IL 62901}
\address{$^b$ School of Computing, Southern Illinois University, 1230 Lincoln Dr, Carbondale, IL 62901}
\ead{quantumdilley@yahoo.com}
\vspace{10pt}
\begin{indented}
\item[]
\date{\today}
\end{indented}

\begin{abstract}
    In open quantum systems, it is known that if the system and environment are in a product state, the evolution of the system is given by a linear completely positive (CP) Hermitian map. CP maps are a subset of general linear Hermitian maps, which also include non completely positive (NCP) maps. NCP maps can arise in evolutions such as non-Markovian evolution, where the CP divisibility of the map (writing the overall evolution as a composition of CP maps) usually fails. Positive but NCP maps are also useful as entanglement witnesses. In this paper, we focus on transforming an initial NCP map to a CP map through composition with the asymmetric depolarizing map. We use separate asymmetric depolarizing maps acting on the individual subsystems.

    Previous work have looked at structural physical approximation (SPA), which is a CP approximation of a NCP map using a mixture of the NCP map with a completely depolarizing map. We prove that the composition can always be made CP without completely depolarizing in any direction. It is possible to depolarize less in some directions. We give the general proof by using the Choi matrix and an isomorphism from a maximally entangled two qudit state to a set of qubits.  We also give measures that describe the amount of disturbance the depolarization introduces to the original map. Given our measures, we show that asymmetric depolarization has many advantages over SPA in preserving the structure of the original NCP map. Finally, we give some examples. For some measures and examples, completely depolarizing (while not necessary) in some directions can give a better approximation than keeping the depolarizing parameters bounded by the required depolarization if symmetric depolarization is used.
\end{abstract}

%
%
%
%
%

\section{Introduction}


Linear Hermitian maps play an important role in many areas of quantum information science. The evolution of an open quantum system is often clearly non-unitary and given by a linear Hermitian map \cite{sudarshan_1961}. For quantum technologies, knowledge of the error map allows one to construct efficient quantum error correction codes provided that some conditions are met \cite{Nielsen_Chuang_Textbook_2011}. Linear Hermitian maps can be categorized as completely positive (CP) or non CP (NCP) (defined below) and the latter can be used as entanglement witnesses \cite{horodcecki1996SepOfMixedStsNecAndSuffCond,Terhal:00}. This paper focuses on transforming NCP maps to CP maps via composition with local asymmetric depolarizing maps. This transformation is important for research in areas such as entanglement, quantum error correction, and non Markovian evolution.

When the system and bath are initially uncorrelated, the system is known to evolve under a completely positive map.  This is a map, say $\Phi$, that is not only positive (maps positive operators to positive operators) but is also positive when extended by an arbitrary identity map $\mathcal{I}$, i.e., $\Phi\rightarrow \mathcal{I}\otimes\Phi$ to act on a larger system. However, it can be quite difficult to prepare the system such that it starts completely uncorrelated from its environment. For such cases, it is not known generally how to describe the system evolution with a CP linear Hermitian map \cite{Pechukas_1994, Alicki_1995, Pechukas_1995, Jordan_Shaji_Sudarshan_2004, Shaji_Sudarshan_2005}. Note, that it is possible in specific cases for the evolution of a system with initial correlations to the environment to be CP \cite{liu2014CPMaps, Rodriguez_Rosario_2008, Modi_2012OperationalApproachToOpenDyn}. It is also claimed in \cite{Jagadish_2021_InitialEntEntUnitariesAndCPMaps} that for any initial non maximally entangled system and environment pure state, there are infinitely many nonlocal unitaries that induce CP evolution on the system. Finally, there is the process tensor formalism \cite{Pollock_2018NonMarkovQuantProccCompleteFrameAndEfficCharac, Modi_2012OperationalApproachToOpenDyn}, which always results in CP linear Hermitian maps, but the process tensor acts on input maps and not on density matrices. In this paper, we are focused on mappings from density matrices to density matrices.

There is also disagreement on the physicality of NCP evolution maps \cite{allen2017Causal, Schmid_Ried_Spekkens_2018}. However, in the case of non-Markovian evolution the initial CP map tends to be non CP divisible (i.e., the map cannot be decomposed into CP maps for all arbitrary intermediate time steps) \cite{Pollock_2018NonMarkovQuantProccCompleteFrameAndEfficCharac}. This implies that an intermediate evolution map can be NCP \cite{Milz_2019CPDivisibilityDoesNotMeanMarkov}. For this paper, we focus on transforming a given NCP map to a CP map and leave discussions of physicality of NCP maps for future work. Still, we require the domain of the NCP map to be restricted so that the output of the map is positive.

Positive but NCP maps do not preserve the positivity of some entangled states, but they do for all separable states. Consequently, they can be used as entanglement witnesses \cite{Horodecki_2009_QuantumEntanglement}. The most common witness is the negativity of the partial transpose which is a necessary and sufficient condition for detecting entanglement for $2\times 2$ or $2\times 3$ systems \cite{Peres,horodcecki1996SepOfMixedStsNecAndSuffCond}. To implement the partial transpose experimentally, Horodecki and Eckert introduced structural physical approximation (SPA), which is a CP approximation of a NCP map \cite{Horodecki_Ekert_2002}.

In the SPA, the completely depolarizing map is mixed with the NCP map. This allows us to make the overall map CP. Note that this is equivalent to composing with a symmetric depolarization map. Symmetric depolarization destroys information by reducing the magnitude of the Bloch vector, or polarization vector, by reducing it equally in all directions.  This preserves the direction of the vector produced by the original NCP map. However, if we only focus on getting an approximation to the map that is CP, in many situations it is often unnecessary and worse to use symmetric depolarization. For example, if the NCP map only changes the state in one direction, we likely only need to depolarize in that direction to get a CP approximation. 

In quantum error correction, the goal is to correct for the effects of the error map. The known necessary and sufficient conditions for quantum error correction require that the error map is CP \cite{knill_1997TheoryOfQECC, bennett_1996MixedStateEntAndQEC}. In the case of NCP errors, the necessary and sufficient conditions to correct a CP map can lead to codes that are not in the domain of the error map, which means that the process, as described, is not physical \cite{gonzales2020sufficientCondAndConsForRevGenQuant}. Thus, the results of the present paper may help to enable the correction of NCP error maps by transforming them via the asymmetric depolarizing map (ADM) into CP maps.


Furthermore, the process of creating a CP composition via the ADM may reduce the amount of errors effecting the system. For example (see Subsection \ref{subsec:ex4Repolar}), if the NCP error map causes a re-polarization error, i.e., it extends the Bloch vector of the input state, the asymmetric depolarizer that creates a CP composition can eliminate the re-polarization error. In these situations, using the ADM is better than the symmetric depolarizer because the ADM offers more freedom for minimizing the errors occurring on the input state.

In this paper, we investigate composing the asymmetric depolarizing map with a NCP map to get a CP composition. We give a novel proof for the complete positivity of the composition by using an isomorphism between a bipartite maximally entangled qudit and its qubit representation. From this isomorphism, the positivity of the Choi matrix is proven. In Section \ref{sec:admmeas}, we determine the amount of disturbance the ADM introduces and compare it with symmetric depolarization. Using the fidelity and our $M_1$ measure (see Eq. \eqref{Eq:M_1_measuregeneral}), we provide examples that show the advantages of using asymmetric depolarizer over the symmetric one.

The $M_1$ provides a measure of the overall depolarization used to modify a map.  In some scenarios, it is beneficial to use the asymmetric depolarizer to depolarize more in some directions than is required when using the symmetric depolarizer. This may sound counter-intuitive, but one may obtain a better fidelity or $M_1$ measure in these instances than would have been attained by restricting depolarization in all directions to be less than or equal to the amount required by symmetric depolarization. We give an example of this case in \ref{subsec:ex3}.

\section{Background}





In the case of an uncorrelated system and bath, the evolution of the system is given by a completely positive (CP) map that can be described by \cite{Choi_1975}
\begin{align}\label{eq:osrCP}
    \mathcal{E}(\rho_S) = \sum_i E_i\rho_S E_i^\dagger.
\end{align}
We drop subscripts when it is clear from context which subspace we are considering.  $\mathcal{E}$ is CP iff $\mathcal{E}$ is a positive map, (it is positive for all positive inputs) and the map $\mathcal{E}\otimes \mathcal{I}_n$ is also positive for all positive integers $n$ \cite{Choi_1975}. Complete positivity can also be described by the Choi matrix. Let $\mathcal{E}:\mathbb{C}^{n\times n}\rightarrow \mathbb{C}^{m\times m}$. Then the map is CP iff its Choi matrix
\begin{align}
    C_\mathcal{E}=\sum_{ij=0}^{n-1}\mathcal{E}(\op{i}{j})\otimes\op{i}{j}
\end{align}
is positive \cite{Choi_1975}.  (See also Sudarshan, Matthews and Rau \cite{sudarshan_1961}, hereafter referred to as (SMR).)

If the map is not completely positive map (NCP)
\begin{align}\label{eq:osrNCP}
    \mathcal{E}(\rho_S) = \sum_i \eta_iE_i\rho_S E_i^\dagger,
\end{align}
where $\eta_i\in\{\pm 1\}$ and it is necessary that there $\exists \eta_i=-1$ \cite{Choi_1975, Pechukas_1994, Alicki_1995, Pechukas_1995, Jordan_Shaji_Sudarshan_2004, Shaji_Sudarshan_2005}.  (This may potentially happen if the system and environment are initially correlated before they evolve.)  

The first representation of maps we gave in Eq.~\eqref{eq:osrCP} and Eq.~\eqref{eq:osrNCP} is called the operator sum representation (OSR).  In SMR, the evolution of the system is given by the A-matrix $\mathcal{A}$ acting on the vectorized form of $\rho_S$. This vectorization is given by
\begin{align}
    \text{vec}(\rho_s)=\text{vec}\left(\sum{c_{ij}\op{\alpha_i}{\beta_j}}\right)=\sum{c_{ij}\ket{\alpha_i\beta_j}},
\end{align}
where $\ket{\alpha_i}$ and $\ket{\beta_j}$ are basis vectors. Then, the action of $\mathcal{A}$ is given by matrix multiplication
\begin{align}
    \text{vec}(\rho^\prime) = \mathcal{A}\text{vec}(\rho),
\end{align}
where we dropped the text vec for simplification of notation. In index notation, this is given by \cite{sudarshan_1961}
\begin{align}\label{eq:smr_def}
    \rho'_{r^\prime s^\prime} = \sum_{rs}\mathcal{A}_{r^\prime s^\prime,rs}\rho_{rs}.
\end{align}
We refer to this as the SMR representation which was introduced in Ref.~\cite{sudarshan_1961}.  To preserve hermiticity and trace, $\mathcal{A}$ also satisfies
\begin{align}
    \mathcal{A}_{s'r',sr}=(\mathcal{A}_{r's',rs})^*,
\end{align}
and
\begin{align}
    \sum_{r'} A_{r'r',sr} = \delta_{sr},
\end{align}
respectively.

There is also a B-matrix $\mathcal{B}$, which is related to $\mathcal{A}$ and given by
\begin{align}
    \mathcal{B}_{r'r,s's}\equiv\mathcal{A}_{r's',rs}.
\end{align}
The constraints on the A-matrix imply conditions on B for hermiticity
\begin{align}
    \mathcal{B}_{r'r,s's}=(\mathcal{B}_{s's,r'r})^*,
\end{align}
and trace preservation
\begin{align}
    \sum_r \mathcal{B}_{r'r,s'r}=\delta_{r's'}.
\end{align}
The B-matrix is often used due to the fact that $\mathcal{B}$ is Hermitian, and therefore has an eigenvector/eigenvalue decomposition
\begin{align}
    \mathcal{B}_{r^\prime r, s^\prime s} \rho_{rs}= \sum_{\alpha} \gamma(\alpha) C^\alpha_{r^\prime r}\rho_{rs} (C^\alpha_{s^\prime s})^*,
\end{align}
where the $C^{(\alpha)}$ are the eigenvectors and $\gamma$ the eigenvalues of $\mathcal{B}$. One may write the map as
\begin{equation}\label{Bmap}
    \Phi(\rho) = \mathcal{B} \rho = \sum_\alpha \eta_\alpha A_\alpha \rho A_\alpha^\dagger,
\end{equation}
where $A_\alpha \equiv \sqrt{|\gamma(\alpha)|}C^\alpha$ so that $\eta_\alpha = \pm 1$.  Thus ${\cal B}$ is positive iff the map is CP.   In other words, the map is completely positive if and only if all $\eta_\alpha=1$. This shows us the relation between the OSR and SMR representation. It is known that the map is completely positive if and only if all $\eta_\alpha=1$.   

\subsection{Structural Physical Approximation}
The problem of converting a NCP map to a CP map has been studied under structural physical approximation (SPA), which was first introduced by Horodecki and Ekert, and expanded upon by several authors \cite{Horodecki_Ekert_2002, horodecki2001limits,  fiurasek2002spa, Alves_Horodecki_Oi_Kwek_Ekert_2003, horodecki2003QuantEnt, Horodecki2003LimitsOp, Augusiak_2011SPA, Lim2012QutritSPA, Bae2017SPAReview, adhikari2018SPA, kumari2019}. SPAs are often used to approximate the partial transpose map and detect entanglement \cite{Chruscinski2009EntWit, korbicz2008SPABreak, chruscinski2010EntWit, Chruscinski2011EntWit, Augusiak_2011SPA, lim2011ExpSPA, wang2011EntWitFromDensMat, lim2011ExperPartTrans, STORMER2013SepStatesSPA, wang2013SPAEntWit, Kalev2013ApproxTrans, Chruciski2014EntWit, Chrusciski2014DispSPAConj, Adhikari2018EntNegat, Bae2019EntDet}. Writing a trace preserving linear Hermitian map as the affine combination of two completely positive trace preserving (CPTP) maps, Jiang et. al. \cite{jiang_2020_PhysicalImplementOfQuantMapsAndItsAppInErrorMitigation} introduced the physical implementable measure, which can be seen as a measure of how well a NCP map can be approximated with a CP map. The measure is zero if and only if the map is CPTP. Similarly, Regula et. al. \cite{Regula_2021_OperationalAppOfTheDiamondNormAndRelatedMeasInQuantifTheNonPhysOfQuantMaps} defined measures quantifying the cost of simulating NCP maps with a mixture of maps. They investigated the relation with the diamond norm and showed that these measures equal the diamond norm when the map being approximated is linear and trace preserving. Finally, De Santis and Giovannetti \cite{DeSantis_2021_MeasNonMarkovViaIncoherentMixWithMarkovDynamics} defined a measure of non-Markovianity based on optimally approximating a non-Markovian map with a Markovian map. The approximation is performed by mixing the non-Markovian map with a minimum necessary amount of a Markovian map.

In SPA, the NCP superoperator $\Phi$ is approximated by
\begin{align}\label{phiMixtureMaxDep}
    \widetilde{\Phi}(\rho)=p\mathcal{L}_\text{comp}(\rho)+(1-p)\Phi(\rho),
\end{align}
where $\mathcal{L}_\text{comp}$ is the completely depolarizing channel, i.e., $\mathcal{L}_\text{comp}(\rho)=\tr(\rho)\mathbb{I}/d$, where $d$ is the dimension of $\rho$, and $0\leq p \leq 1$. Notice that SPA is equal to the composition
\begin{align}\label{eq:spaComp}
    \widetilde{\Phi}=\mathcal{L}_{1-p}\circ\Phi,
\end{align}
where the $\mathcal{L}_{1-p}$ is the symmetric depolarization channel which scales the polarization vector $n$ \cite{Mahler:book,Jakob:01,Byrd/Khaneja:03,Kimura} 
by the factor $1-p$. The symmetric depolarization channel for a single qubit can be written in OSR as
\begin{align}
    \mathcal{L}_{1-p}(\rho)=\left(1-\dfrac{3p}{4}\right)\rho+\dfrac{p}{4}(\sigma_x\rho\sigma_x+\sigma_y\rho\sigma_y+\sigma_z\rho\sigma_z).
\end{align}

We study the composition
\begin{align}\label{eq:asymComp}
    \widetilde{\Phi'}=\mathcal{L}'\circ\Phi,
\end{align}
where $\mathcal{L}'$ performs asymmetric depolarization on each local state.  The asymmetric depolarization channel for a single qubit can be written in OSR as
\begin{align}
    \label{Eq:OSR_of_ADM}\mathcal{L}(\rho)=&\dfrac{1}{4}\left[(1+\alpha+\beta+\gamma)\rho+(1+\alpha-\beta-\gamma)\sigma_x\rho\sigma_x\right.\\
    &\left.+(1-\alpha+\beta-\gamma)\sigma_y\rho\sigma_y+(1-\alpha-\beta+\gamma)\sigma_z\rho\sigma_z\right],
\end{align}
where $\alpha$, $\beta$, and $\gamma$ are the amounts of depolarization in the $x$, $y$, and $z$ directions, respectively. The asymmetric depolarizer does not generally preserve the direction of the Bloch vector of the input state. From the SPA, we immediately see that $\widetilde{\Phi'}$ can be made CP with appropriate (symmetric) depolarization. However, in the SPA it is necessary that the systems are depolarized symmetrically and equally. Thus, it is not immediately clear that depolarizing the individual local systems differently leads to a CP composition. 

Using an isomorphism between a $2^n$ dimensional qudit and $n$ qubits, we prove that symmetric depolarization is not necessary. We also show that when using asymmetric depolarization, completely depolarizing is not necessary in any direction for any of the local systems for Eq. \eqref{eq:asymComp} to be CP. However, we also show that there exist NCP maps that are robust against depolarization and completely depolarizing is \textit{almost} necessary. Also keep in mind that extending the A-matrix to higher dimensions is not as straight-forward as it may seem. For instance, if we wanted wanted to implement a single qubit channel on the first part of an entangled two-qubit state, the corresponding A-matrix acting on the quantum state will in general not have the form $A \otimes \mathbb{I} (\rho)$. In \ref{Sec:extendingA-matrix}, we give the explicit transformation that makes this extension possible and easy to use. Then we provide a simple example in \ref{Sec:extendingA-matrix_example} to caution readers so they do not make the mistake of using the wrong form of the A-matrix for the extension.


\section{NCP to CP}
Our goal of transforming a NCP map to a CP map is equivalent to going from an initial B-matrix with at least one negative eigenvalue to a final B-matrix with non-negative eigenvalues. Equivalently, this is going from Eq. \eqref{eq:osrNCP} to Eq. \eqref{eq:osrCP} in OSR.

\subsection{Composition: Asymmetric Depolarization}
We use composition with the asymmetric depolarization map, which is also known as generalized depolarization. Obviously, if we completely depolarize the composition would be a CP map because we always end up with the maximally mixed state. Asymmetric depolarization has been studied in other contexts in the past \cite{bowen2001TeleportationDepolarization, Jagadish_2019MeasurePauliMaps, Jagadish_2018InvitationChannels, siudzi2019GeometryPauliMaps}, but here we use it to transform a NCP map to a CP map without completely depolarizing. Note that the ADM is doubly stochastic, i.e., it preserves the identity matrix and it is trace preserving. Before proving the general case, we show an isomorphism from maximally entangled qudits to the basis of tensored Pauli matrices that we will use. We also show this isomorphism to the basis of tensored $3\times 3$ Gell-Mann matrices. This Theorem \ref{thm:maxquditToPauli} will be used to rewrite the Choi matrix in Pauli matrix form.
\begin{theorem}\label{thm:maxquditToPauli}
    For the unnormalized maximally entangled two qudit state with qudit dimension $d=2^{n}$, 
    we have the relation
    \begin{align}\label{eq:maxquditToPauli}
        \notag\op{\widetilde\Phi^+_{d}}=\sum_{i,j=0}^{d-1}\op{ii}{jj}&=\dfrac{1}{d}\sum_{i_1,i_2,...,i_{n}=0}^{3}w_{i_1,i_2,...,i_{n}}(\sigma_{i_1,i_2,...,i_{n}})^{\otimes 2}\\
        &=\dfrac{1}{d}\sum_{\overline{i}=0}^{3}w_{\overline i}(\sigma_{\overline{i}})^{\otimes 2},
    \end{align}
    where $i$ represents the identity or $x, y$ or $z$ for the Pauli matrix, the subscript on the $i$ represents the subspace, $\overline i=i_1,i_2,...,i_{n}$, $(\sigma_{i_1,i_2,...,i_{n}})^{\otimes 2}=(\sigma_{i_1}\otimes \sigma_{i_2}\otimes...\otimes\sigma_{i_{n}})^{\otimes 2}$, and 
    \begin{align}
        \notag w_{i_1,i_2,...,i_{n}}=\{&+1\text{ for an even number of $\sigma_2$s in $\sigma_{i_1,i_2,...,i_{n}}$}\\
                &-1 \text{ for an odd number of $\sigma_2$s in $\sigma_{i_1,i_2,...,i_{n}}$}\}.
    \end{align}
    We refer to the right hand side of Eq. \eqref{eq:maxquditToPauli} as the Pauli Form. Note that $\sigma_0=\mathbb{I}, \sigma_1=\sigma_x$, $\sigma_2=\sigma_y$, and $\sigma_3=\sigma_z$.
    
    For an unnormalized maximally entangled two qudit state with qudit dimension $d=3^{n}$, we have the relation 
    \begin{align}\label{eq:maxquditToGellman}
        \op{\widetilde\Phi^+_{d}}=\sum_{i,j=0}^{d-1}\op{ii}{jj}&=\sum_{\overline{i}=\overline 0}^{\overline 8}w_{\overline i}(\lambda_{\overline{i}})^{\otimes 2},
    \end{align}
    where $\{\lambda_i\}$ are re-scaled $3\times 3$ Gell-Mann matrices including identity (the identity is scaled by a factor of $1/\sqrt{3}$ and the others are scaled by a factor of $1/\sqrt{2}$); $\overline i=i_1,i_2,...,i_{n}$; and
    \begin{align}
        \notag w_{i_1,i_2,...,i_{n}}=\{&+1\text{ for an even number of complex $\lambda$'s in $\lambda_{\overline i}$}\\
                &-1 \text{ for an odd number of complex $\lambda$'s in $\lambda_{\overline i}$}\}.
    \end{align}
\end{theorem}
\begin{proof}
    The proof relies on figuring out the spectral decomposition of the Pauli Form. The ricochet property gives us the useful relation
    \begin{align}
        \sum_i M\otimes \mathbb{I} \ket{ii}=\sum_i \mathbb{I} \otimes M^T\ket{ii}.
    \end{align}
    Acting on the right of the Pauli Form with $\ket{\widetilde{\Phi}^+_d},$ we get
    \begin{align}
        \notag\dfrac{1}{2^{n}}&\sum_{i_1,i_2,...,i_{n}=0}^{3}(w_{i_1,i_2,...,i_{n}}\sigma_{i_1,i_2,...,i_{n}}\otimes\sigma_{i_1,i_2,...,i_{n}})\sum_{i=0}^{d-1}\ket{ii} \\
        \notag&=\dfrac{1}{2^{n}}\sum_{i_1,i_2,...,i_{n}=0}^{3}(w_{i_1,i_2,...,i_{n}}\mathbb{I}\otimes\sigma_{i_1,i_2,...,i_{n}}\sigma_{i_1,i_2,...,i_{n}}^T)\sum_{i=0}^{d-1}\ket{ii}\\
        \notag&=\dfrac{1}{2^{n}}\sum_{i_1,i_2,...,i_{n}=0}^{3}(\mathbb{I}\otimes\mathbb{I})\sum_{i=0}^{d-1}\ket{ii}\\
        \notag&=\dfrac{1}{2^{n}}(4^{n})\sum_{i=0}^{d-1}\ket{ii}\\
        &=2^{n}\sum_{i=0}^{d-1}\ket{ii}.
    \end{align}
    Note that the $\pm$ signs in the definition of $w$ comes from the relation $\sigma_y^T = - \sigma_y$. Therefore, $\ket{\widetilde\Phi^+_d}$ is an eigenvector of the Pauli form. Also, keep in mind that $\ip{\widetilde\Phi^+}=d=2^{n}$. Thus, we can write our Pauli form in its spectral decomposition as
    \begin{align}\label{eq:quditPauli1}
        \dfrac{1}{d}\sum_{i_1,i_2,...,i_{n}=0}^{3}w_{i_1,i_2,...,i_{n}}(\sigma_{i_1,i_2,...,i_{n}})^{\otimes 2}=\op{\widetilde\Phi_d^+}+\sum_i\lambda_i\op{\lambda_i},
    \end{align}
    where $\lambda_i$ are real eigenvalues and $\{\ket{\lambda_i},\ket{\widetilde\Phi^+_d}\}$ are orthogonal eigenvectors. Next, we show that the Pauli Form is a pure state and thus the $\lambda_i$ are all zero. Taking the square of both sides of Eq. \eqref{eq:quditPauli1} and then taking the trace, we have
    \begin{align}\label{eq:quditPauli2}
        \tr\left(\left[\dfrac{1}{d}\sum_{i_1,i_2,...,i_{n}=0}^{3}w_{i_1,i_2,...,i_{n}}(\sigma_{i_1,i_2,...,i_{n}})^{\otimes 2}
        \right]^2\right)&=\tr\left(\left[\op{\widetilde\Phi_d^+}+\sum_i\lambda_i\op{\lambda_i}\right]^2\right).
    \end{align}
    Since the Pauli matrices are traceless, Eq. \eqref{eq:quditPauli2} simplifies to
    \begin{align}
        \notag\Rightarrow\notag\tr\left[\dfrac{1}{d^2}\sum_{i_1,i_2,...,i_{n}=0}^{3}\mathbb{I}\otimes\mathbb{I}
        \right]&=\tr\left[2^{n}\op{\widetilde\Phi_d^+}+\sum_i\lambda_i^2\op{\lambda_i}\right]\\
        \notag\Rightarrow\dfrac{1}{4^{n}}4^{n}4^{n}&=4^{n}+\sum_i\lambda_i^2\\
        \notag\Rightarrow 4^n&=4^n+\sum_i\lambda_i^2\\
        \Rightarrow \lambda_i&=0.
    \end{align}
    Therefore, we must have
    \begin{align}
        \dfrac{1}{d}\sum_{i_1,i_2,...,i_{n}=0}^{3}w_{i_1,i_2,...,i_{n}}(\sigma_{i_1,i_2,...,i_{n}})^{\otimes 2}=\op{\widetilde\Phi_d^+}.
    \end{align}
    
    The qutrit case follows the same arguments. We refer to the right hand side of Eq. \eqref{eq:maxquditToGellman} as the Gell-Mann form. Once again acting on the right hand side of Eq. \eqref{eq:maxquditToGellman} with $\ket{\Phi_d^+}=\sum_{j=0}^{d-1}\ket{jj}$ and using the ``ricochet" property we get
    \begin{align}
       \notag\sum_{\overline{i}=\overline 0}^{\overline 8}w_{\overline i}(\lambda_{\overline{i}})^{\otimes 2}\sum_{j=0}^{d-1}\ket{jj}&=\sum_{\overline{i}=\overline 0}^{\overline 8}(\mathbb{I}\otimes\lambda_{\overline{i}}^2)\sum_{j=0}^{d-1}\ket{jj}\\
       &=3^n\sum_{j=0}^{d-1}\ket{jj}, 
    \end{align}
    where the second line comes from the fact that $\sum_{i=0}^{i=8}\lambda_i^2=3\mathbb{I}$. Therefore, $\sum_{j=0}^{d-1}\ket{jj}$ is an eigenvector of the Gell-Mann form. Thus,
    \begin{align}\label{eq:quditGellman1}
        \sum_{\overline{i}=\overline 0}^{\overline 8}w_{\overline i}(\lambda_{\overline{i}})^{\otimes 2}=\op{\widetilde\Phi_d^+}+\sum_i\gamma_i\op{\gamma_i},
    \end{align}
    where $\gamma_i$ are real eigenvalues and $\{\ket{\gamma_i},\ket{\widetilde\Phi^+_d}\}$ are orthogonal eigenvectors. Finally, we prove that the Gell-Mann form is a pure state. Squaring both sides of Eq. \eqref{eq:quditGellman1} we get
    \begin{align}\label{eq:quditGellman2}
        \notag&\tr\left[\left(\sum_{\overline{i}=\overline 0}^{\overline 8}w_{\overline i}(\lambda_{\overline{i}})^{\otimes 2}\right)^2\right]=\tr\left[\left(\op{\widetilde\Phi_d^+}+\sum_i\gamma_i\op{\gamma_i}\right)^2\right]\\
        \notag&\rightarrow\tr\left[\sum_{\overline{i}=\overline 0}^{\overline 8}(\lambda_{\overline{i}}\otimes \lambda_{\overline{i}})^2+\sum_{\overline{i}\neq \overline{j}}w_{\overline i}(\lambda_{\overline{i}})^{\otimes 2}*w_{\overline j}(\lambda_{\overline{j}})^{\otimes 2}\right]\\
        &=\tr\left[3^{n}\op{\widetilde\Phi_d^+}+\sum_i\gamma_i^2\op{\gamma_i}\right],
    \end{align}
    where on the second line we explicitly separated the cross terms on the left hand side. Note that $\tr\left(\lambda_i^2\otimes\lambda_i^2\right)=1$ $\forall i.$
    
    Using this property and continuing Eq. \eqref{eq:quditGellman2}
    \begin{align}\label{eq:quditGellman3}
        \notag&\rightarrow 3^{2n} + \tr\left[\sum_{\overline{i}\neq \overline{j}}w_{\overline i}(\lambda_{\overline{i}})^{\otimes 2}*w_{\overline j}(\lambda_{\overline{j}})^{\otimes 2}\right]\\
        \notag&=\tr\left(\left[3^{n}\op{\widetilde\Phi_d^+}+\sum_i\gamma_i^2\op{\gamma_i}\right]^2\right)\\
        \notag&\rightarrow3^{2n}=3^{2n}+\sum_i\gamma_i^2\\
        &\rightarrow\sum_i\gamma_i^2=0,
    \end{align}
    where we used the fact that the trace of the cross terms on the left hand side is zero and $\sum_i\ip{ii}=3^{n}$. Thus, the Gell-Mann form is a pure state and the result follows.
\end{proof}

Theorem \ref{thm:maxquditToPauli} has an interesting consequence that is shown in \ref{App:2}. We give the explicit relation between k-pairs of maximally entangled qubits and a maximally entangled qudit. We can now use the help of Theorem \ref{thm:maxquditToPauli} to prove our main theorem.
\begin{theorem}\label{thm:NCP_to_CP}
Any $n$ qubit trace-preserving NCP map $\mathcal{A}_\text{NCP}$ can be made CP by composing it with a map that performs asymmetric depolarization on each qubit with appropriate nonzero values of $\alpha_i, \beta_i,$ and $\gamma_i$, where $i$ represents the $i^\text{th}$ qubit being depolarized. 
\end{theorem}
\begin{proof}
We again make use of the Choi matrix. For a $n$ qubit system, we perform the ADM on each qubit individually. This action is equivalent to performing
\begin{align}
    \mathcal{L}'=\Pi_{i=1}^n\mathcal{L}_i\otimes \mathcal{I},
\end{align}
where $i$ represents the $i^\text{th}$ qubit. We can maximally depolarize by letting $\alpha_i=\beta_i=\gamma_i=0$. This depolarizing map is completely positive when the following inequalities are satisfied:
\begin{align}
    |\gamma_i \pm \alpha_i| &\leq 1 \pm \beta_i.
\end{align}
The action of $\mathcal{L}_i$ on the Pauli matrices are
    \begin{align}\label{eq:asymPauli}
        \mathcal{L}_i(\sigma_x)=\alpha_i\sigma_x,\qquad \mathcal{L}_i(\sigma_y)=\beta_i\sigma_y, \qquad\text{and}\qquad \mathcal{L}_i(\sigma_z)=\gamma_i\sigma_z.
    \end{align}
The identity matrix is left invariant. $\mathcal{A}$ is a mapping from $n$ qubits to $m$ qubits ($m$ may or may not equal $n$). Thus, $\mathcal{A}: \mathbb{C}^{2^{n}\times 2^{n}}\rightarrow \mathbb{C}^{2^m\times 2^m}$. Let $d=2^n$. The Choi matrix for the composition is
\begin{align}\label{eq:generalCPproof0}
    \notag\sum_{i=0}^{d-1}\mathcal{L'}\circ\mathcal{A}(\op{i}{j})\otimes\op{i}{j}&=\dfrac{1}{d}\sum_{i_1,i_2,...,i_{n}=\overline{0}}^{\overline{3}}w_{i_1,i_2,...,i_{n}}\mathcal{L'}\circ\mathcal{A}(\sigma_{i_1,i_2,...,i_{n}})\otimes\sigma_{i_1,i_2,...,i_{n}}\\
    \notag &=\dfrac{1}{d}\sum_{\overline{i}=\overline{0}}^{\overline{3}} w_{\overline i}\mathcal{L'}\circ\mathcal{A}(\sigma_{\overline i})\otimes\sigma_{\overline i}\\
    &=\dfrac{1}{d}\left[\mathcal{L'}\circ\mathcal{A}(\mathbb{I})\otimes\mathbb{I}+\sum_{\overline{i}\neq\overline{0}}^{\overline{3}} w_{\overline{i}}\mathcal{L'}\circ\mathcal{A}(\sigma_{\overline{i}})\otimes\sigma_{\overline{i}}\right],
\end{align}
where $\overline{0}=00...0$, $\overline i=i_1,i_2,...,i_{n}$, we used Theorem \ref{thm:maxquditToPauli} on the right hand side of the first line, and we explicitly separated the identity term on the last line. Like before, since $\mathcal{A}$ is trace preserving, and Hermitian matrices can be expanded in the basis of tensor of Pauli matrices including identity,
\begin{align}\label{eq:generalCPproof1}
        \mathcal{A}(\mathbb{I})=k\left(\mathbb{I}+\sum_{\overline{j}\neq \overline{0}}^{\overline{3}}(a_{\overline{j}}\sigma_{\overline{j}})_0\right)
    \end{align}
    and
    \begin{align}\label{eq:generalCPproof2}
        \mathcal{A}(\sigma_{\overline{i}})=\left(\sum_{\overline{j}\neq \overline{0}}^{\overline{3}}(a_{\overline{j}}\sigma_{\overline{j}})_{\overline{i}}\right),\qquad \forall \overline{i}\neq \overline{0},
    \end{align}
where $k=2^n/2^m$ is a positive nonzero constant to preserve the trace ($\mathcal{A}$ is trace preserving) and the $a_{\overline{j}}$'s are arbitrary constants. Substituting Eq. \eqref{eq:asymPauli}, Eq. \eqref{eq:generalCPproof1}, and Eq. \eqref{eq:generalCPproof2} into Eq. \eqref{eq:generalCPproof0}, we have
\begin{align}\label{eq:generalCPproof3}
    \notag\sum_{i=0}^{2^n-1}\mathcal{L'}\circ\mathcal{A}&(\op{i}{j})\otimes\op{i}{j}=\dfrac{1}{d}\left[\mathcal{L'}\circ\mathcal{A}(\mathbb{I})\otimes\mathbb{I}+\sum_{\overline{i}\neq\overline{0}}^{\overline{3}} w_{\overline{i}}\mathcal{L'}\circ\mathcal{A}(\sigma_{\overline{i}})\otimes\sigma_{\overline{i}}\right]\\
    \notag =&\dfrac{k}{d}\left[\mathcal{L'}\left[\mathbb{I}+\sum_{\overline{j}\neq \overline{0}}^{\overline{3}}(a_{\overline{j}}\sigma_{\overline{j}})_{\overline{0}}\right]\otimes\mathbb{I}+\dfrac{1}{k}\sum_{\overline{i}\neq \overline{0}}\left(w_{\overline{i}}\mathcal{L'}\left[\sum_{\overline{j}\neq \overline{0}}^{\overline{3}}(a_{\overline{j}}\sigma_{\overline{j}})_{\overline{i}}\right]\otimes\sigma_{\overline{i}}\right)\right]\\
    \notag=&\dfrac{k}{d}\left[\left(\mathbb{I}+\sum_{\overline{j}\neq \overline{0}}^{\overline{3}}\left[m_{\overline{j}}(a_{\overline{j}}\sigma_{\overline{j}})_{\overline{0}}\right]\right)\otimes\mathbb{I}+\dfrac{1}{k}\sum_{\overline{i}\neq \overline{0}}\left(w_{\overline{i}}\sum_{\overline{j}\neq \overline{0}}^{\overline{3}}\left[m_{\overline{j}}(a_{\overline{j}}\sigma_{\overline{j}})_{\overline{i}}\right]\otimes\sigma_{\overline{i}}\right)\right]\\
    =&\dfrac{k}{d}(\mathbb{I}\otimes\mathbb{I}+J),
\end{align}
where the $m_{\overline{j}}$'s are the products of the appropriate $\alpha_i, \beta_i$ and $\gamma_i$ constants from the local ADMs and $J$ is everything in between the square braces that is not $\mathbb{I}\otimes\mathbb{I}$. The Choi matrix Eq. \eqref{eq:generalCPproof3} becomes positive as $J$ goes to zero. We can arbitrarily shrink $J$ asymmetrically by decreasing $\alpha_i, \beta_i,$ and $\gamma_i$ independently. Thus, from continuity, with nonzero values of $\alpha_i, \beta_i,$ and $\gamma_i$, $\mathcal{L}'\circ\mathcal{A}$ is CP. 
\end{proof}
\begin{remark}
    The results also hold for NCP trace preserving qutrit maps. The Choi matrix for the composition can be written using $3\times 3$ Gell-Mann matrices from Theorem \ref{thm:maxquditToPauli} and the results follow from the same line of argument as the qubit case.
\end{remark}

A corollary follows directly from Theorem \ref{thm:NCP_to_CP} for single qubit maps. First we will give the explicit construction of the asymmetric depolarizing map for single qubits:
\begin{align}
    \label{Eq:asymmetric_depolarizing_map}\mathcal{L} &= 
    \op{\Phi^+}{\Phi^+} + \alpha \op{\Psi^+}{\Psi^+} + \beta \op{\Psi^-}{\Psi^-} + \gamma \op{\Phi^-}{\Phi^-}\\ 
    =&
    \dfrac{1}{2}
    \begin{bmatrix}
    1+\gamma &0 &0 &1-\gamma\\
    0 &\alpha+\beta &\alpha-\beta &0\\
    0 &\alpha-\beta &\alpha+\beta &0\\
    1-\gamma &0 &0 &1+\gamma
    \end{bmatrix} \qquad \text{for } \alpha, \beta, \gamma \in [-1,1].
\end{align}

\begin{corollary}
Let the single qubit asymmetric depolarization map be given by Eq. (\ref{Eq:asymmetric_depolarizing_map}) which has the OSR in Eq. (\ref{Eq:OSR_of_ADM}). Then let $\mathcal{A}$ be a single qubit trace preserving NCP map. The composition $\mathcal{L}\circ\mathcal{A}$ is CP for appropriate nonzero values of $\alpha, \beta,$ and $\gamma.$
\end{corollary}
\begin{proof}
This depolarizing map is completely positive when the following \textit{Fujiwara-Algoet} \cite{Fujiwara_1999OneToOneParamOfQuantChannels} conditions are satisfied:
\begin{align}
    |\gamma \pm \alpha| &\leq 1 \pm \beta.
\end{align}
We can maximally depolarize by letting $\alpha=\beta=\gamma=0$. We use the Choi matrix representation. The action of $\mathcal{L}$ on the Pauli matrices are
    \begin{align}\label{eq:singleAsymmPauli}
        \mathcal{L}(\sigma_x)=\alpha\sigma_x,\qquad \mathcal{L}(\sigma_y)=\beta\sigma_y, \qquad\text{and}\qquad \mathcal{L}(\sigma_z)=\gamma\sigma_z.
    \end{align}
    The ADM leaves the identity matrix invariant. Next, note that the unnormalized maximally entangled 2 qubit state $\op{\widetilde{\Phi}^+}$ is given by
    \begin{align}
        \op{\widetilde{\Phi}^+}=\dfrac{1}{2}(\mathbb{I}\otimes\mathbb{I}+\sigma_x\otimes\sigma_x-\sigma_y\otimes\sigma_y+\sigma_z\otimes\sigma_z).
    \end{align}
    We can now get to the main proof. The Choi matrix for $\mathcal{L}\circ\mathcal{A}$ is given by
    \begin{align}\label{eq:2qubitComp}
        \notag\sum_{i,j=0}^1\mathcal{L}\circ\mathcal{A}(\op{i}{j})\otimes\op{i}{j}=&\dfrac{1}{2}[\mathcal{L}\circ\mathcal{A}(\mathbb{I})\otimes\mathbb{I}+\mathcal{L}\circ\mathcal{A}(\sigma_x)\otimes\sigma_x\\
        &-\mathcal{L}\circ\mathcal{A}(\sigma_y)\otimes\sigma_y+\mathcal{L}\circ\mathcal{A}(\sigma_z)\otimes\sigma_z].
    \end{align}
    We then use the fact that $\mathcal{A}$ is trace preserving, the Pauli matrices are Hermitian, and a Hermitian matrix can be expanded in the basis of all possible tensors of Pauli matrices. We have
    \begin{align}\label{eq:2qubitproof1}
        \mathcal{A}(\mathbb{I})=\mathbb{I}+\sum_{k=1}^3(a_k\sigma_k)_0
    \end{align}
    and
    \begin{align}\label{eq:2qubitproof2}
        \mathcal{A}(\sigma_j)=\sum_{k=1}^3(a_k\sigma_k)_j,\qquad \forall j\in\{x,y,z\},
    \end{align}
    where the $a_k$'s are arbitrary constants. Substituting Eq. \eqref{eq:singleAsymmPauli}, Eq. \eqref{eq:2qubitproof1}, and Eq. \eqref{eq:2qubitproof2} into Eq. \eqref{eq:2qubitComp} we have
    \begin{align}\label{eq:2qubitproof3}
        \notag\sum_{i,j=0}^1\mathcal{L}\circ\mathcal{A}(\op{i}{j})\otimes\op{i}{j}=&\dfrac{1}{2}\left(\mathcal{L}\left[\mathbb{I}+\sum_{k=1}^3(a_k\sigma_k)_0\right]\otimes\mathbb{I}+\mathcal{L}\left[\sum_{k=1}^3(a_k\sigma_k)_x\right]\otimes\sigma_x\right.\\
        \notag &\left.-\mathcal{L}\left[\sum_{k=1}^3(a_k\sigma_k)_y\right]\otimes\sigma_y+\mathcal{L}\left[\sum_{k=1}^3(a_k\sigma_k)_z\right]\otimes\sigma_z \right)\\
        \notag=&\dfrac{1}{2}\left[\mathbb{I}\otimes\mathbb{I}+\sum_{k=1}^3m_k(a_k\sigma_k)_0\otimes\mathbb{I}+\sum_{k=1}^3m_k(a_k\sigma_k)_x\otimes\sigma_x\right.\\
        \notag&\left. -\sum_{k=1}^3m_k(a_k\sigma_k)_y\otimes\sigma_y+\sum_{k=1}^3m_k(a_k\sigma_k)_z\otimes\sigma_z\right]\\
        &=\dfrac{1}{2}(\mathbb{I}\otimes\mathbb{I}+J),
    \end{align}
where $m_k$ corresponds to $\alpha, \beta,$ or $\gamma$ and $J$ is everything inside the square braces except the $\mathbb{I}\otimes \mathbb{I}$ term. The Choi matrix, Eq. \eqref{eq:2qubitproof3}, becomes positive as $J$ goes to zero. We can arbitrarily shrink $J$ asymmetrically by decreasing the values of $\alpha, \beta, $ and $\gamma$ independently. Thus, from continuity, with nonzero values of $\alpha, \beta, $ and $\gamma$, $\mathcal{L}\circ\mathcal{A}$ is CP for some choice of $\alpha, \beta, $ and $\gamma$.
\end{proof}

It has been shown that the volume of NCP maps is twice the volume as CP maps for Pauli channels when subjected to a full positivity domain \cite{Jagadish_2019MeasurePauliMaps}. There may be valid A-matrices that have an associated B-matrix that is Hermitian and trace 2, such that there is an intersection between the initial domain and the image space, but the physical process may be impossible to achieve. In \ref{App:3}, we show that such an example may exist. There is always an intersection between the image space and the initial domain of the Bloch sphere, but the A-matrix has the extreme property of mapping the resultant state's eigenvalues of invalid initial density operators toward $\pm \infty$. This suggests that no physical process can accomplish such a task even though it satisfies all the properties of being a valid NCP mapping.

\subsection{Asymmetric Depolarization Map Measures}\label{sec:admmeas}
From this point on, we let $\mathcal{L}=\Pi_{i=1}^n\mathcal{L}_i\otimes \mathcal{I}$ or the single qubit ADM. When we perform ADM, we generally change the state's Bloch vector. Therefore, it is important to measure how much the state is changed. There are multiple measures we can use. First, we give our ADM measure $M_1$
\begin{align}
\label{Eq:M_1_measure}
    M_1:=\dfrac{1}{3}(\abs{\alpha}+\abs{\beta}+\abs{\gamma}),
\end{align}
which is basically just an averaging over the absolute values of the depolarizing parameters. The range of $M_1$ is $0\le M_1\le 1$. At the two extremes, we have total loss of the Bloch vector and the magnitude of the Bloch vector unchanged for the values of $0$ and $1$, respectively. For $n$ qubits, this measure generalizes in the obvious way to
\begin{align}\label{Eq:M_1_measuregeneral}
    M_1:=\dfrac{1}{3n}\left(\sum_{i=1}^{3n}\abs{c_i}\right),
\end{align}
where the $c_i$'s are the depolarization parameters. The $M_1$ measure for the symmetric depolarizing channel is simply given by a single depolarizing parameter that we can call $\tau$ that is between $0$ and $1$.

For single qubits, we also have the closed form for measures using fidelity and the diamond norm. We can calculate the fidelity between the initial state and the final state for the asymmetric depolarizer when it is composed with some error map $\mathcal{A}$. Define the matrix
\begin{align}
    D =
    \left(
    \begin{array}{ccc}
       \alpha & 0 & 0 \\
        0 & \beta & 0 \\
        0 & 0 & \gamma
    \end{array}
    \right).
\end{align}
We can define two fidelity measures. The first fidelity is
\begin{align}\label{eq:fidelityErrorCorr}
    F(\rho, \mathcal{L} \circ \mathcal{A} \rho) = \dfrac{1}{2}\bigg\{1+ \vec{r}\cdot D \vec{r} \;' + \sqrt{(1-\vec{r} \cdot \vec{r})(1-\vec{r} \; ' \cdot D^2 \vec{r} \; ')} \bigg\},
\end{align}
where $\vec{r}$ is the Bloch vector of the initial state $\rho$ and $\vec{r} \;'$ is the Bloch vector of the state $\mathcal{A} \rho$. This fidelity is good for error correction since we could minimize the effect of the composition while ensuring complete positivity of the composition. For the symmetric depolarizer, $D=\tau I$. If we are interested in approximating the $A$-matrix while ensuring a CP composition with the asymmetric depolarizer, then we want to use the calculation
\begin{align}\label{eq:fidelityApproxA}
    F(\mathcal{A} \rho, \mathcal{L} \circ \mathcal{A} \rho) = \dfrac{1}{2}\bigg\{1+ \vec{r} \; '\cdot D \vec{r} \;' + \sqrt{(1-\vec{r} \; ' \cdot \vec{r} \; ')(1-\vec{r} \; ' \cdot D^2 \vec{r}\;')} \bigg\}.
\end{align}

When comparing the gate fidelity of two Pauli channels, the general equation is given by the diamond norm \cite{Benenti_2010} and has the analytical form of
\begin{align}
|| \mathcal{E}_1 (\rho) - \mathcal{E}_2 (\rho) ||_\diamond = \sum_i | \kappa^{(1)}_i - \kappa^{(2)}_i |
\end{align}
where each channel has the form
\begin{align}
    \mathcal{E}_j (\rho) = \sum_i \kappa^{(j)}_i \hat{\sigma}_i \rho \hat{\sigma}_i \text{ such that } \sum_i \kappa^{(j)}_i = 1 \text{ for } j = 1,2.
\end{align}
In the case where $\mathcal{E}_1 (\rho)$ is the asymmetric depolarizing channel and $\mathcal{E}_2 (\rho)$ is the symmetric depolarizing channel, we have that the diamond norm $|| \mathcal{E}_1 (\rho) - \mathcal{E}_2 (\rho) ||_\diamond$ is equal to
\begin{align}\label{eq:diamondNormAsymWithSym}
    \dfrac{1}{4} \bigg\{ |\alpha+\beta+\gamma-3\tau|+|\alpha-\beta-\gamma+\tau|+|\alpha-\beta+\gamma-\tau|+|\alpha+\beta-\gamma-\tau| \bigg\}.
\end{align}
The parameters $\alpha, \beta,$ and $\gamma$ describe the asymmetric depolarizer and the parameter $\tau$ describes the symmetric one. The distance between the asymmetric depolarizer and symmetric depolarizer gives us a measure of the advantage of the asymmetric depolarizer because the performance of the asymmetric depolarizer is generally lower bounded by the symmetric depolarizer.



\subsection{ADM Composition vs SPA: Information Loss}
We would like to show that the composition with the asymmetric depolarizer is generally better than the SPA at retaining information. As was shown in the introduction, the SPA is equivalent to composition with a symmetric depolarizer. The natural way to compare the two compositions would be to use mutual information and show that $I(\mathcal{L}\circ\mathcal{A}(\rho);\mathcal{A}(\rho))\geq I(\mathcal{L}_\text{symm}\circ\mathcal{A}(\rho);\mathcal{A}(\rho))$, where $\mathcal{L}$ is the asymmetric depolarizer, $\mathcal{L}_\text{symm}$ is the symmetric depolarizer, $\mathcal{A}$ is the input NCP A-matrix, and $I(A;B)$ is the mutual information for the systems $A$ and $B$. However, since this calculates the mutual information for product states, i.e., $A$ and $B$ are not correlated, these values are zero. Instead, we bound the linear entropy $S_L(\rho)=1-\tr(\rho^2)$ \cite{Peters2004MixedStateSensitivity}.
\begin{lemma}
    Let $\mathcal{A}$ be a trace preserving NCP map; $\tau$ be the depolarization parameter for $\mathcal{L}_\text{symm}$; $c_i$'s be the depolarization parameters for $\mathcal{L}$; and $\mathcal{L}\circ\mathcal{A}$ and $\mathcal{L}_\text{symm}\circ\mathcal{A}$ be CP with the minimum depolarization required (minimum depolarization is given by maximizing $M_1$ with $c_i\geq \tau$ $\forall i$). Then, the linear entropy is bounded by
    \begin{align}
        S_L(\mathcal{L}_\text{symm}\circ\mathcal{A}(\rho))\geq S_L(\mathcal{L}\circ\mathcal{A}(\rho))\quad \forall \rho.
    \end{align}
\end{lemma}
\begin{remark}
    If the depolarization parameters are \textit{not} constrained by $c_i\geq \tau$ $\forall i$, the result does not hold as shown by the counter examples in Subsection \ref{subsec:ex3}. Also, note that the linear entropy is a good measure for qubits but can be a misleading measure for higher dimensional states due to the possible distribution of the eigenvalues.
\end{remark}
\begin{proof}
    The minimum depolarization for the ADM composition $\mathcal{L}\circ\mathcal{A}$ can be bounded by going back to the positivity of the Choi matrix of the compositions. 
    Notice that the Choi matrix for the SPA is equivalent to the Choi matrix for the ADM composition when all the depolarizing parameters are equal (this must be true because when the depolarizing parameters are all equal we are symmetrically depolarizing). Then, the composition is CP for some optimal symmetric depolarization value $c$. At this value, the purity of the output of the two compositions are the same and the symmetric depolarization used for the SPA is optimal. Since $c_i\geq \tau$ $\forall i$, the purity of the output of the ADM composition must be greater than or equal to the purity of the output of the SPA and the bound follows.
\end{proof}

\subsection{Composition with ADM as a Structural Physical Approximation}
Can the composition $\mathcal{L}\circ\mathcal{A}$ be considered a structural physical approximation of the NCP A-matrix? Structural physical approximations should preserve the structure of the original A-matrix by preserving the direction of the Bloch vector of the output of the original NCP A-matrix \cite{horodecki2001limits}. The ADM does not in general preserve the direction of the Bloch vector. When the ADM depolarizes symmetrically, the composition is completely equivalent to the SPA if local depolarizers are also used for the SPA.

\section{Examples}
We give some examples of NCP maps and their corresponding ADMs that convert them to CP maps. The first two examples are detailed and for the last two examples we just give the corresponding ADM. Note that we let $\mathcal{L}=\mathcal{L}(\alpha,\beta,\gamma)$ and $\mathcal{T}=\mathcal{T}(x_0,y_0,z_0)$ (the translation matrix defined in Eq. \eqref{eq:genTranslation}).

\subsection{Example 1}
This example illustrates the advantage of asymmetric over symmetric depolarization with regards to making CP maps. In general, a single qubit trace-preserving map has the form \cite{sudarshan_1961}
\begin{align}
    A = 
    \left[
    \begin{array}{cccc}
     a & b & b^* & d  \\
     e & f & g & h \\
     e^* & g^* & f^* & h^* \\
     1-a & -b & -b^* & 1-d
    \end{array}
    \right].
\end{align}
For this example, let us use the NCP translation matrix
\begin{align}
    \mathcal{T}(0, 0, 1/2) = \dfrac{1}{4}
    \left[
    \begin{array}{cccc}
     5 & 0 & 0 & 1 \\
     0 & 4 & 0 & 0 \\
     0 & 0 & 4 & 0 \\
     -1 & 0 & 0 & 3 
    \end{array}
    \right]
\end{align}
so that the Bloch sphere gets shifted up from its initial position along the $z$-axis by $1/2$ as seen in Fig. (\ref{Fig:NCP_to_CP_via_asym}). 
Intuitively, the asymmetric depolarizer that will work best (as given by our measures in Section \ref{sec:admmeas}), for making this a valid CP map by composition, is given by shrinking the $z$ component of the Bloch vector by $2/3$. Then, the $\alpha$ and $\beta$ components are $\sqrt{2/3}$. Note that we get a CP composition when we depolarize the $z$ component by $2/3$ and $\alpha,\beta \in [0,\sqrt{2/3}]$. We obviously would choose $\alpha=\beta=\sqrt{2/3}$ because we would like to depolarize by the least amount possible. 

The symmetric depolarizer that will work best, for making this a valid CP map by composition, is given by shrinking the $z$ component of the Bloch vector by $2/3$ which also shrinks the $x$ and $y$ components by the same proportion. From our $M_1$ measure we can see that the ADM composition performs  better than the SPA. $M_1=2/3$ for the symmetric depolarizer compared to $M_1=(2+2\sqrt{6})/9 \approx 0.767 > 0.667$ for the ADM. In terms of the Bloch sphere, the ADM shrinks the image space of the transpose into an ellipsoid, while the symmetric depolarizer retains the image shape of a sphere. The volume of the image space for the ADM composition is always greater than the volume of the SPA image space. However, the symmetric depolarizer maintains the shape of the image space of the original NCP map, but the ADM does not.

The eigenvalues of the resultant dynamical B-matrix for
\begin{align}
    \mathcal{L}(\sqrt{2/3}, \sqrt{2/3}, 2/3) \circ \mathcal{T}(0, 0, 1/2)
\end{align}
are $\{ 5/3, 1/3, 0, 0 \}$ and for 
\begin{align}
    \mathcal{L}(\sqrt{2/3}, \sqrt{2/3}, 2/3)
\end{align} are $\{ (1/6)(5+2 \sqrt{6}), 1/6, 1/6, (1/6)(5-2 \sqrt{6}) \}$. 
For the SPA, the eigenvalues of the resultant dynamical B-matrix
\begin{align}
    \mathcal{L}(2/3, 2/3, 2/3) \circ \mathcal{T}(0, 0, 1/2)
\end{align} are $\{ (1/6)(5+\sqrt{17}), 1/3, (1/6)(5-\sqrt{17}),0 \}$ and for
\begin{align}
    \mathcal{L}(2/3, 2/3, 2/3)
\end{align} are $\{ (3/2, 1/6, 1/6, 1/6) \}$.
\begin{figure}[h]
    \centering
    \includegraphics[scale=.5]{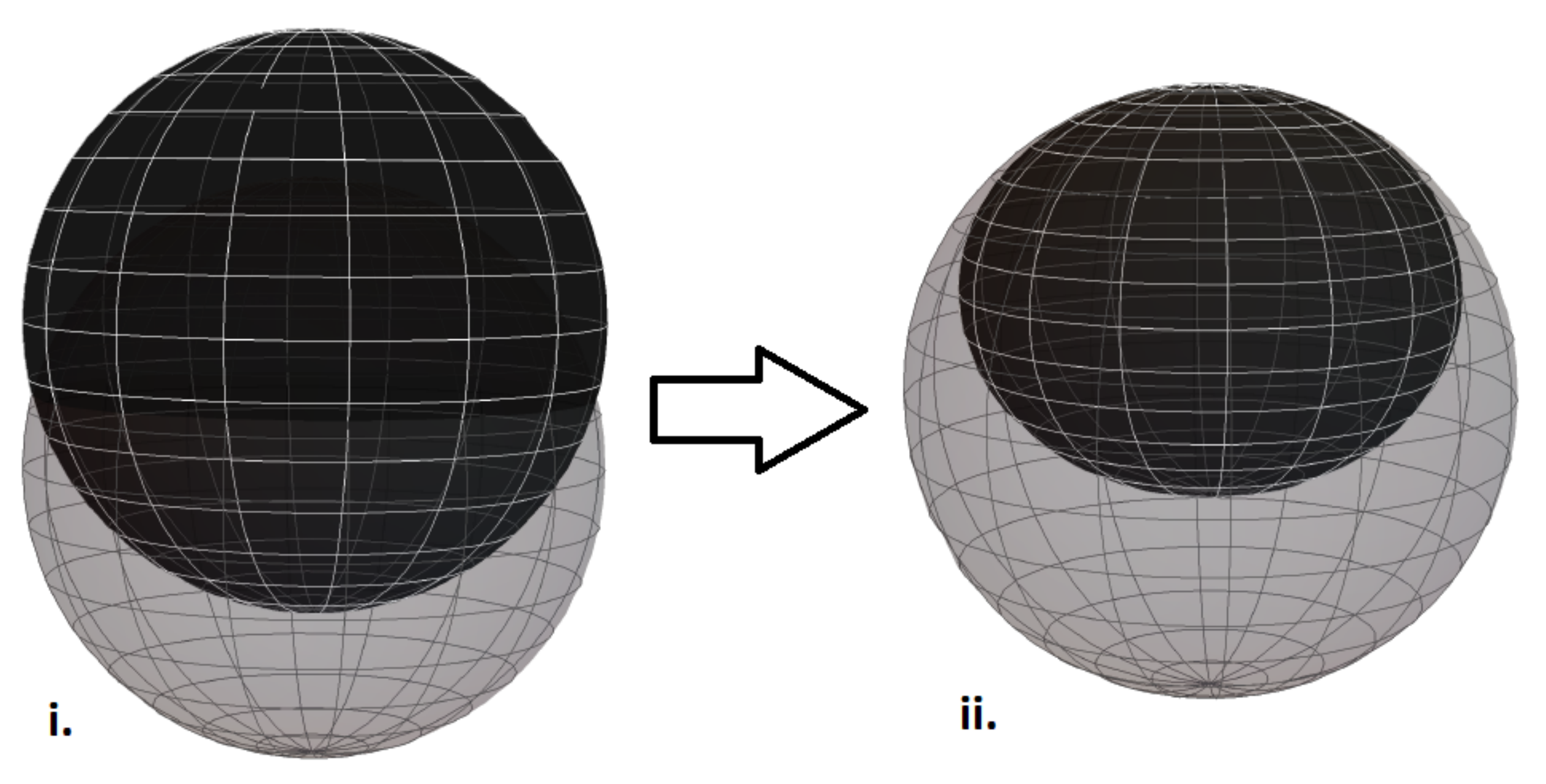}
    \caption{The first picture illustrates the action of the NCP translation map on an arbitrary qubit on the Bloch sphere. The lighter image represents the original Bloch sphere and the darker sphere is the image of the translation map along the $z$-axis. The second image demonstrates how the ellipsoid positivity domain of the final state becomes completely inscribed inside the original Bloch sphere after the asymmetric depolarizing map is composed with the translation map. Therefore, we have a final CP map.}
    \label{Fig:NCP_to_CP_via_asym}
\end{figure} 

Notice how both the composition map and also the $\mathcal{L}$ map must have non-negative eigenvalues to have a valid implementable protocol. In general, the minimum degree of depolarization can be obtained in a numerical optimization program that maximizes the $M_1$ measure in Eq. (\ref{Eq:M_1_measure}) while leaving both sets of eigenvalues non-negative. A calculation of the fidelity in Eq. (\ref{eq:fidelityErrorCorr}) between initial pure states and final states for the ADM composition and the SPA are given by
\begin{align}
    \mathcal{F}_{\text{ADM}}(\rho, \mathcal{L} \circ \mathcal{T}(\rho)) = \dfrac{1}{12} \left( 8 + \sqrt{6} + 2 \; \text{cos}(\theta) + (2 - \sqrt{6}) \; \text{cos}(2 \theta) \right),
\end{align}
and
\begin{align}
    \mathcal{F}_{\text{SPA}}(\rho, \mathcal{L} \circ \mathcal{T}(\rho))  = \dfrac{1}{6} \bigg( 5 + \text{cos}(\theta) \bigg)
\end{align}
which are only dependent on the polar angle $\theta$. The difference of the ADM composition and SPA fidelity measures is given by $(1/6)(\sqrt{6}-2)\text{sin}(\theta)^2$ which is always greater than or equal to zero. Thus, the ADM is better at state preservation than the symmetric depolarizer as illustrated in Figure (\ref{Fig:Fidelity_NCP_to_CP}).  
\begin{figure}[h!]
    \centering
    \includegraphics[scale=.7]{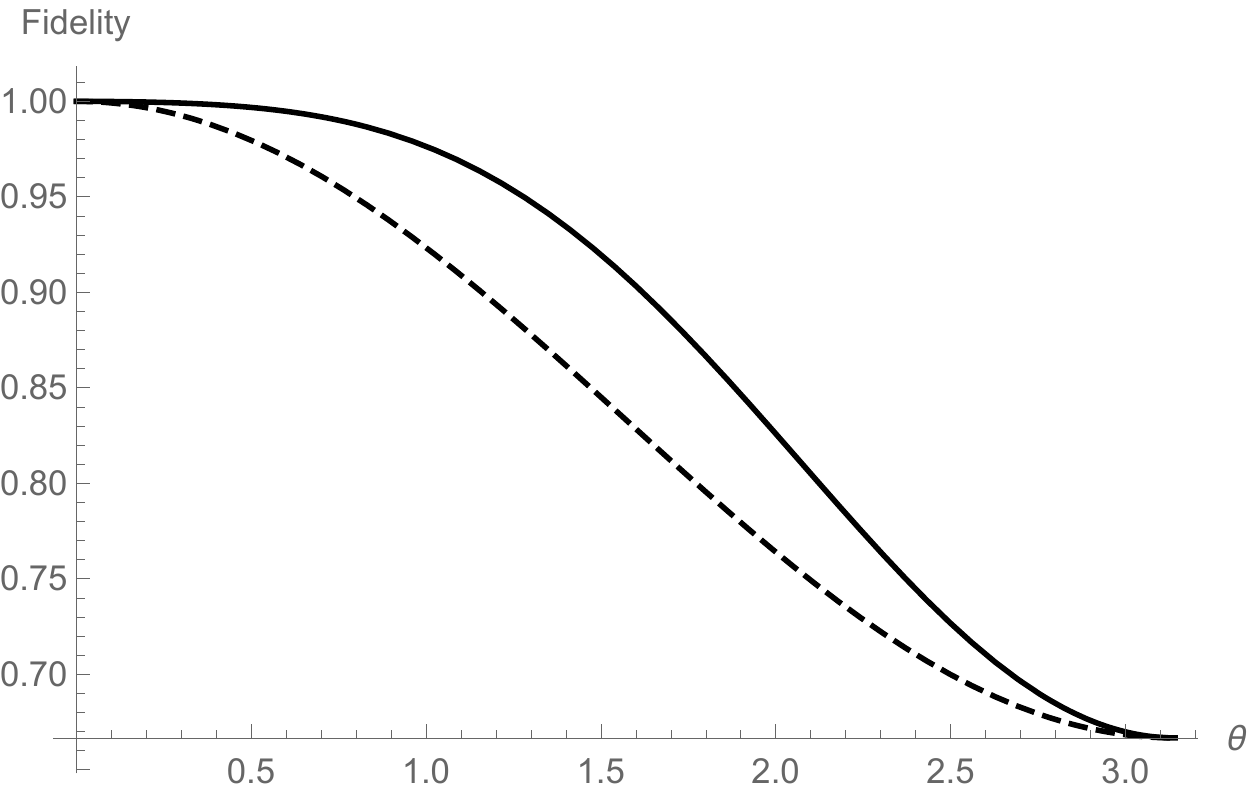}
    \caption{The solid line gives the fidelity, with respect to the parameter $\theta$, between the output of the ADM composition with the initial state, when the initial states are pure. The dashed line represents the fidelity for the SPA with the initial state. The lowest achievable fidelity in each case is 2/3 which occurs at $\theta = \pi$, but for any other value of $\theta$ corresponding to non-unit fidelity, we see that the ADM composition has a higher fidelity than the SPA.}
    \label{Fig:Fidelity_NCP_to_CP}
\end{figure}
The gate fidelity in Eq. \eqref{eq:diamondNormAsymWithSym} between the two depolarizers is $(1/3)(\sqrt{6}-2) \approx .15$. Interesting enough, this value is double the maximum difference of the fidelity of the best ADM composition and SPA set-ups. \\

\section{Conclusion}
In this paper, we studied transforming a NCP map to a CP map through composition with the asymmetric depolarizing map (ADM). The ADM acts as a super-superoperator that maps the input NCP map to a CP map. This problem is similar to structural physical approximation. However, in SPA, symmetric depolarization is required. We found that a trace preserving NCP map on qubits can always be made CP by acting on the individual qubits with local ADMs. The proof uses an isomorphism between a maximally entangled $2$ qudit state and $n$ qubits. Furthermore, we prove that these local ADMs never have to completely depolarize in any direction and are not required to be symmetric.

Note that a global ADM can preserve the structure of the input NCP map better than local ADMs on the individual qubits. This is due to the fact that the depolarization factors from the multiple local depolarizers multiply on the correlation matrix. Thus, the correlation matrix shrinks more rapidly than the local Bloch, or polarization, vectors. However, experimentally implementing a global ADM is very difficult for a large number of qubits. Next, we defined measures of the ADM, which tells us how much the input state into the ADM is disturbed. Under these measures, the ADM composition has advantages over SPA in preserving the structure of the original map.

Next, we gave examples. One interesting example (in Eq. \eqref{Eq:A_Robust}), shows that there exist valid A-matrices that are robust against asymmetric depolarization and completely depolarizing becomes \text{almost} necessary. We also showed with another example (in Eq. \eqref{eq:Ex:maxDepolInSomeDirs}) that, under the measure we called $M_1$ and also the fidelity, sometimes completely depolarizing in one direction (while not necessary) results in a better approximation. For qutrit maps, the task of transforming a NCP map to a CP map via local ADMs is unproven. Also, we did not consider the positivity domain since our method did not rely on it. In general, an induced NCP map will have a valid domain in which the subsystem can be contained in \cite{Rodriguez_Rosario_2008,Dominy_2016,Vacchini_2016}. We leave these topics for future work.

Finally, note that in some cases, the composition of the initial NCP map with an asymmetric depolarizing map may or may not be CP divisible depending on the strength of the depolarizer. This can be seen from the extremes. With no depolarizing, the composition is equal to the initial NCP map, which is trivially not CP divisible. For NCP maps that commute with the ADM, when we depolarize completely, the composition is equivalent to the completely depolarizing map, which is CP divisible. From continuity, there must be some depolarization values that cause the transition from not CP divisible to CP divisible. This transition may be useful as a measure of non-Markovianity. This is similar to the measure of non-Markovianity defined by De Santis and Giovannetti \cite{DeSantis_2021_MeasNonMarkovViaIncoherentMixWithMarkovDynamics}, but their measure is based on a mixture of an initial non-Markovian map with a Markovian map. We leave this for future work.

\section*{Acknowledgments}

Funding for this research was provided by the NSF, MPS under award number PHYS-1820870.  The authors thank Purva Thakre, Eric Chitambar, and Dario De Santis for helpful comments.  


\appendix

\section{}

\subsection{Extending the A-matrix to Higher Dimensions}
\label{Sec:extendingA-matrix}

One might think that a valid A-matrix $\mathcal{A}$ is trivially extensible to higher dimensions in the SMR representation by $\mathcal{A}\otimes \mathcal{I}$. However, this is not true as we show by a simple counter example below. If we want the tensor product of two A-matrices $\mathcal{A}_1 \otimes \mathcal{A}_2$ to have the correct action when operating on a two-qubit density operator, we need to focus our attention to the operator-sum representation. By definition, when an A-matrix acts on the vectorized form of a single qubit density operator, the resultant state is the same thing we would have gotten by sending the qubit through a channel expressed in OSR; that is,
\begin{align}\label{Eq:Qubit_vectoized}
    (\mathcal{A} \rho_v)_d = \sum_i E_i \rho E_i = \mathcal{E}(\rho),
\end{align}
where the subscripts $d$ and $v$ represent density and vector forms of an operator, respectively. An arbitrary two-qubit state has the form
\begin{align}
    \rho^{AB} = \dfrac{1}{4} \bigg\{ \mathbb{I} \otimes \mathbb{I} + \vec{r} \cdot \vec{\sigma} \otimes \mathbb{I} + \mathbb{I} \otimes \vec{s} \cdot \vec{\sigma} + \sum_{i,j=1}^3 t_{ij} \hat{\sigma}_i \otimes \hat{\sigma}_j \bigg\},
\end{align}
which when acted upon by $\mathcal{A}_1 \otimes \mathcal{A}_2$ becomes
\begin{align}
    &\rho^{AB} = (\mathcal{A}_1 \otimes \mathcal{A}_2) \dfrac{1}{4} \bigg\{ (\mathbb{I})_v \otimes (\mathbb{I})_v + (\vec{r} \cdot \vec{\sigma})_v \otimes (\mathbb{I} )_v+ (\mathbb{I})_v \otimes (\vec{s} \cdot \vec{\sigma})_v + \sum_{i,j=1}^3 t_{ij} (\hat{\sigma}_i)_v \otimes (\hat{\sigma}_j)_v \bigg\} \\
    &= \dfrac{1}{4} \bigg\{ (\mathbb{I})_v \otimes (\mathbb{I})_v + \mathcal{A}_1 (\vec{r} \cdot \vec{\sigma})_v \otimes (\mathbb{I} )_v+ (\mathbb{I})_v \otimes \mathcal{A}_2 (\vec{s} \cdot \vec{\sigma})_v + \sum_{i,j=1}^3 t_{ij} \mathcal{A}_1 (\hat{\sigma}_i)_v \otimes \mathcal{A}_2 (\hat{\sigma}_j)_v \bigg\}. \nonumber
\end{align}

The reason we vectorized all of the $2 \times 2$ sub-matrices in this state is that it simulates Eq. (\ref{Eq:Qubit_vectoized}) exactly on the subsystems. Thus, it is the same as if we derived the operator-sum decomposition for $\mathcal{A}_1$ and $\mathcal{A}_2$ and then acted upon the state $\rho_{AB}$ with $\mathcal{E}_1 (.) \otimes \mathcal{E}_2(.)$. A really nice structure occurs when calculating the resultant vectorized state for a two-qubit system. It turns out that if our state initially has the representation 
\begin{align}
    \rho^{AB} = 
    \left(
    \begin{array}{cc}
        A & B \\
        C & D
    \end{array}
    \right),
\end{align}
then the correct vectorized form of it is given by
\begin{align}
\label{Eq:New_vectorization}
\rho^{AB}_v =
\left(
\begin{array}{c}
     A_v \\
     B_v \\
     C_v \\
     D_v
\end{array}
\right).
\end{align}
Therefore, we can use the normal extension $\mathcal{A}_1 \otimes \mathcal{A}_2$ if we vectorize our two-qubit state locally.

Generally, we can simply perform a rotation on some general A-matrix $\mathcal{W}$ so that its rows act on $\rho_v$ the same way $\mathcal{A}_1 \otimes \mathcal{A}_2$ does. This equates to the exchange of the rows $3 \leftrightarrow 5, 4 \leftrightarrow 6, 11 \leftrightarrow 13, \text{and } 12 \leftrightarrow 14$, which is essentially a permutation matrix of the following form:
\begin{align}
    \label{Eq:permutation}
    R =
    \left(
    \begin{array}{cccccccccccccccc}
        1 & 0 & 0 & 0 & 0 & 0 & 0 & 0 & 0 & 0 & 0 & 0 & 0 & 0 & 0 & 0 \\
        0 & 1 & 0 & 0 & 0 & 0 & 0 & 0 & 0 & 0 & 0 & 0 & 0 & 0 & 0 & 0 \\
        0 & 0 & 0 & 0 & 1 & 0 & 0 & 0 & 0 & 0 & 0 & 0 & 0 & 0 & 0 & 0 \\
        0 & 0 & 0 & 0 & 0 & 1 & 0 & 0 & 0 & 0 & 0 & 0 & 0 & 0 & 0 & 0 \\
        0 & 0 & 1 & 0 & 0 & 0 & 0 & 0 & 0 & 0 & 0 & 0 & 0 & 0 & 0 & 0 \\
        0 & 0 & 0 & 1 & 0 & 0 & 0 & 0 & 0 & 0 & 0 & 0 & 0 & 0 & 0 & 0 \\
        0 & 0 & 0 & 0 & 0 & 0 & 1 & 0 & 0 & 0 & 0 & 0 & 0 & 0 & 0 & 0 \\
        0 & 0 & 0 & 0 & 0 & 0 & 0 & 1 & 0 & 0 & 0 & 0 & 0 & 0 & 0 & 0 \\
        0 & 0 & 0 & 0 & 0 & 0 & 0 & 0 & 1 & 0 & 0 & 0 & 0 & 0 & 0 & 0 \\
        0 & 0 & 0 & 0 & 0 & 0 & 0 & 0 & 0 & 1 & 0 & 0 & 0 & 0 & 0 & 0 \\
        0 & 0 & 0 & 0 & 0 & 0 & 0 & 0 & 0 & 0 & 0 & 0 & 1 & 0 & 0 & 0 \\
        0 & 0 & 0 & 0 & 0 & 0 & 0 & 0 & 0 & 0 & 0 & 0 & 0 & 1 & 0 & 0 \\
        0 & 0 & 0 & 0 & 0 & 0 & 0 & 0 & 0 & 0 & 1 & 0 & 0 & 0 & 0 & 0 \\
        0 & 0 & 0 & 0 & 0 & 0 & 0 & 0 & 0 & 0 & 0 & 1 & 0 & 0 & 0 & 0 \\
        0 & 0 & 0 & 0 & 0 & 0 & 0 & 0 & 0 & 0 & 0 & 0 & 0 & 0 & 1 & 0 \\
        0 & 0 & 0 & 0 & 0 & 0 & 0 & 0 & 0 & 0 & 0 & 0 & 0 & 0 & 0 & 1 
    \end{array}
    \right).
\end{align}
Therefore, if we want to use the same vectorization, specifically the one given in Eq. (\ref{Eq:New_vectorization}), for $\rho_v$ when acted on by $\mathcal{W}$, we simply apply the rotation $R$ such that $\mathcal{W} \rightarrow R \mathcal{W} R^T$. By doing this, we can always use the vectorization in Eq. (\ref{Eq:New_vectorization}) to obtain the correct result. For any amount of qubits $d$, there exists a permutation matrix that will always lead to the same vectorization method for any $\mathcal{A}_1 \otimes \mathcal{A}_2 \otimes \cdots \otimes \mathcal{A}_d$. 
Hence, the two-qubit A-matrix of $\mathcal{A}_1 \otimes \mathcal{A}_2$ is given by $R \mathcal{A}_1 \otimes \mathcal{A}_2 R$. Note that for this particular permutation matrix $R = R^T$ and $R^2 = \mathbb{I}$.

The following example extends a single qubit A-matrix with identity to act on a two-qubit state.

\section{Some Cautionary Examples}
\label{Sec:extendingA-matrix_example}

Let our initial two qubit state be the maximally mixed state
\begin{align}
    \rho_{AB}=\dfrac{1}{4}\mathbb{I}_{AB}.
\end{align}
Let us perform the maximal depolarization on system $A$ and leave the system $B$ alone. We might naturally assume that the global A-matrix can be calculated directly as
\begin{align}
    \mathcal{A}=\mathcal{L}_\text{comp}\otimes\mathcal{I}=
    \dfrac{1}{2}\begin{pmatrix}
    1 & 0 & 0 & 1\\
    0 & 0 & 0 & 0\\
    0 & 0 & 0 & 0\\
    1 & 0 & 0 & 1
    \end{pmatrix} \otimes
    \mathbb{I}.
\end{align}
Then straight-forward calculation gives
\begin{align}
    \notag\mathcal{A}\text{vec}(\rho_{AB})=\dfrac{1}{4}\text{vec}\left(\Phi^+\right).
\end{align}

This result is obviously wrong. The correct output state is $\dfrac{1}{4}\mathbb{I}$. To find the correct A-matrix for the extension to higher dimensions, we can use the OSR. Let $\mathcal{E}_\mathcal{L}(\rho)=\sum_iE_i\rho E_i^\dagger$ correspond to the OSR of $\mathcal{L}_{\text{comp}}$. Then, the extended map is given by
\begin{align}
    \mathcal{E}_\mathcal{L}\otimes\mathcal{I}(\rho)=\sum_i(E_i\otimes \mathbb{I})\rho(E_i^\dagger\otimes \mathbb{I}).
\end{align}
Using Eq. \eqref{Bmap} we can return to the SMR representation and get the B-matrix and the A-matrix. 
The correct B and A-matrices are
\begin{align}
    \notag\mathcal{B}&=1/4\sum_i{\text{vec}(\sigma_i\otimes\mathbb{I})\text{vec}(\sigma_i\otimes\mathbb{I})^\dagger}\\
    \notag&\Rightarrow\\
    \mathcal{A}&=\left(\begin{array}{cccccccccccccccc}
         1/2 & 0 & 0 & 0 & 0 & 0 & 0 & 0 & 0 & 0 & 1/2 & 0 & 0 & 0 & 0 & 0\\
         0 & 1/2 & 0 & 0 & 0 & 0 & 0 & 0 & 0 & 0 & 0 & 1/2 & 0 & 0 & 0 & 0\\
         0 & 0 & 0 & 0 & 0 & 0 & 0 & 0 & 0 & 0 & 0 & 0 & 0 & 0 & 0 & 0\\
         0 & 0 & 0 & 0 & 0 & 0 & 0 & 0 & 0 & 0 & 0 & 0 & 0 & 0 & 0 & 0\\
         0 & 0 & 0 & 0 & 1/2 & 0 & 0 & 0 & 0 & 0 & 0 & 0 & 0 & 0 & 1/2 & 0\\
         0 & 0 & 0 & 0 & 0 & 1/2 & 0 & 0 & 0 & 0 & 0 & 0 & 0 & 0 & 0 & 1/2\\
         0 & 0 & 0 & 0 & 0 & 0 & 0 & 0 & 0 & 0 & 0 & 0 & 0 & 0 & 0 & 0\\
         0 & 0 & 0 & 0 & 0 & 0 & 0 & 0 & 0 & 0 & 0 & 0 & 0 & 0 & 0 & 0\\
         0 & 0 & 0 & 0 & 0 & 0 & 0 & 0 & 0 & 0 & 0 & 0 & 0 & 0 & 0 & 0\\
         0 & 0 & 0 & 0 & 0 & 0 & 0 & 0 & 0 & 0 & 0 & 0 & 0 & 0 & 0 & 0\\
         1/2 & 0 & 0 & 0 & 0 & 0 & 0 & 0 & 0 & 0 & 1/2 & 0 & 0 & 0 & 0 & 0\\
         0 & 1/2 & 0 & 0 & 0 & 0 & 0 & 0 & 0 & 0 & 0 & 1/2 & 0 & 0 & 0 & 0\\
         0 & 0 & 0 & 0 & 0 & 0 & 0 & 0 & 0 & 0 & 0 & 0 & 0 & 0 & 0 & 0\\
         0 & 0 & 0 & 0 & 0 & 0 & 0 & 0 & 0 & 0 & 0 & 0 & 0 & 0 & 0 & 0\\
         0 & 0 & 0 & 0 & 1/2 & 0 & 0 & 0 & 0 & 0 & 0 & 0 & 0 & 0 & 1/2 & 0\\
         0 & 0 & 0 & 0 & 0 & 1/2 & 0 & 0 & 0 & 0 & 0 & 0 & 0 & 0 & 0 & 1/2
    \end{array}\right).
\end{align}

\subsection{Some Valid A-matrices May Not Exist Physically}
\label{App:3}

Keep in mind that there exist NCP maps that are robust to asymmetric and symmetric depolarizers. This means that the parameters $\alpha, \beta,$ and $\gamma$ must be extremely close to zero in order to obtain a CP composition with the NCP A-matrix. For instance, if we define a NCP map to be
\begin{align}
\label{Eq:A_Robust}
    \mathcal{A} = 
    \left(
    \begin{array}{cccc}
        \kappa & \kappa & \kappa & \kappa \\
        \kappa & 0 & \kappa & 0 \\
        \kappa & \kappa & 0 & 0 \\
        1-\kappa & -\kappa & -\kappa & 1-\kappa 
    \end{array}
    \right),
\end{align}
where $\kappa>1/2(3-\sqrt{5})$, we see that the parameters of the asymmetric depolarizer must be less than or equal to $1/\kappa$ to get a CP composition $\mathcal{L} \circ \mathcal{A}$. To see this, we completely depolarize in two directions and determine the constraint on the third direction. The eigenvalues of the composition with $\mathcal{L}(1/\kappa, 0, 0)$ and $\mathcal{L}(0, 0, 1/\kappa)$ contain one negative value and the composition with $\mathcal{L}(0,\beta,0)$ requires $\beta=1/\kappa$ for a CP composition, which gives the eigenvalues of $\{ 1, 1, 0, 0 \}$ for the composition. Therefore, maps where $\kappa$ becomes increasingly large certainly do not capture any advantages over the symmetric depolarizer. We can look at the individual directions separately because we are depolarizing asymmetrically.

Since there exists initial states outside the domain whose eigenvalues of their resultant states blow up to $\pm \infty$ as $\kappa \rightarrow \infty$, we need to check the validity of this map. 
Just like the asymmetric depolarizer, we can define a translation of the original state's Bloch vector with the A-matrix
\begin{align}\label{eq:genTranslation}
    \mathcal{T} = \dfrac{1}{2} 
    \left(
    \begin{array}{cccc}
        2+z_0 & 0 & 0 & z_0 \\
        x_0 - i y_0 & 2 & 0 & x_0 - i y_0 \\
        x_0 + i y_0 & 0 & 2 & x_0 + i y_0 \\
        -z_0 & 0 & 0 & 2-z_0
    \end{array}
    \right)
\end{align}
to perform the action
\begin{align}
    \{ x, y, z \} \longrightarrow \{ x+x_0, y+y_0, z+z_0 \}
\end{align}
on the initial Bloch vector \cite{Jagadish_2019_General_Case}. This transformation will obviously be NCP for any non-zero translation vector $\{ x_0, y_0, z_0 \}$ since the image Bloch sphere is not fully contained in the initial Bloch Sphere. This translation matrix only has a valid physical domain when the translation vector has a magnitude less than or equal to $2$, even though it is a valid A-matrix for values outside of this interval. This begs the question of whether or not (\ref{Eq:A_Robust}) represents a valid mapping of our qubit. The answer to this question is yes, it is valid. Figure \ref{fig:ncpInfMap} illustrates the action of this map when $\kappa=1$. 

When $\kappa \rightarrow \infty$, we see that the state
\begin{align}
\label{Eq:One_point}
    \rho = \dfrac{1}{2} \big( \mathbb{I} - \sigma_x \big) \longrightarrow A(\rho) = \dfrac{1}{2} \big( \mathbb{I} - \sigma_z \big) 
\end{align}
always leads to a valid state.
The domain for the output state increases as $\kappa$ decreases. In the limit as $\kappa\rightarrow\infty$, the only valid output is the $\ket{1}$-state given by the transformation in (\ref{Eq:One_point}). So this valid A-matrix can have an operator-sum decomposition with eigenvalues, in the limit, that blow up to $\pm \infty$ while illustrating true robustness to the asymmetric depolarizing map. It is very interesting how there exists a mapping that is so resistant to heavy noise in all directions, yet the magnitude of the Bloch vector for the $\ket{-}$-state is left invariant at the end of that mapping. Future research can look into the physicality of this map such as the global evolution on some correlated state that would induce such a transformation on one of the subsystems.

\begin{figure}[h]
    \centering
    \includegraphics[scale=.7]{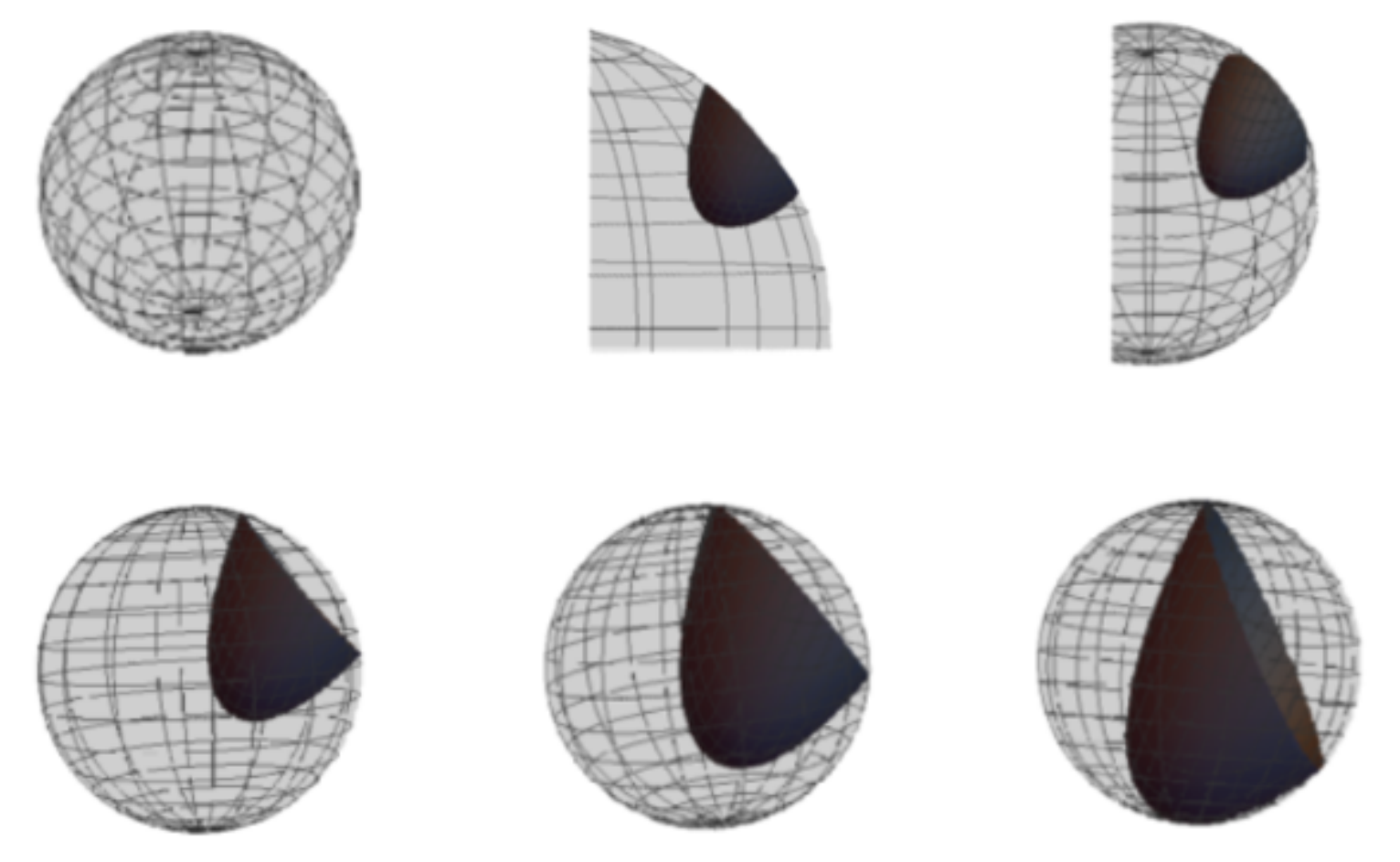}
    \caption{On top left, the length of the initial Bloch vector is $r=1/6$ and it increases in length by intervals of $1/6$ from left to right. So the bottom right Bloch sphere has a radius of 1. We fix $\kappa = 1$. We see that as the length of the initial Bloch vector increases, there is more intersection between the Bloch sphere and its image. Thus, we only need to consider when $r=1$ to determine that the transformation is always possible with increasing $\kappa$.}
    \label{fig:ncpInfMap}
\end{figure}

\section{Interesting Relation}
\label{App:2}

    We can find a unitary relation between $k$ pairs of maximally entangled qubits (ebits) $\ket{\widetilde{\Phi}^+_2}^{\otimes k}$ and a maximally entangled bipartite qudit state $\ket{\widetilde\Phi^+_d}$. Specifically,
    \begin{align}\label{eq:qubitsToQuditsURelat1}
        U\ket{\widetilde{\Phi}^+_2}^{\otimes k}=\ket{\widetilde\Phi^+_d},
    \end{align}
    where $U$ is a SWAP operator and $d=2^{k}$. We only need to consider $k\geq2$. Notice that $\ket{\widetilde\Phi^+_d}=\sum_{i=0}^{d-1}\ket{ii}$ is isomorphic to the identity $\sum_{i=0}^{d-1}\op{i}$. When there is an even number of ebits, Eq. \eqref{eq:qubitsToQuditsURelat1} implies
    \begin{align}
        U_{\text{even}} \sum_{i\in \{ 0, 1 \}} \op{i_1,i_1, i_2, i_2,\cdots, i_{k/2}, i_{k/2}}{i_{k/2+1},i_{k/2+1},i_{k/2+2},i_{k/2+2},\cdots, i_{k}, i_{k}}=\mathbb{I}
    \end{align}
    and when there is an odd number of ebits,
    \begin{align}
    \label{Eq:n_not_divisible_by_4}
        U_{\text{odd}} \sum_{i \in \{ 0, 1 \}} \op{i_1,i_1,i_2, i_2,\cdots, i_{\lfloor k/2 \rfloor +1}}{i_{\lfloor k/2 \rfloor +1},i_{\lfloor k/2 \rfloor +2},i_{\lfloor k/2 \rfloor +2},\cdots, i_{k}, i_{k}}= \mathbb{I},
    \end{align}
    where the subscript $k$ represents the $k^\text{th}$ ebit pair and the two situations require a different form of SWAP. Keep in mind that $U$ is done before applying the isomorphism to the state $\ket{\widetilde{\Phi}^+}^{\otimes k}$ so we can swap any of the systems. The idea is to make the string in the ket and the bra equal. Let $n=2k$ be the number of qubits and the qubits in the outer product be labeled as
    \begin{align}
        \op{1, 2, \cdots, \dfrac{n}{2}}{\dfrac{n}{2}+1, \cdots, n-1, n}.
    \end{align}
    
    For an even number of ebits,
    \begin{align}
        \label{Eq:n_divisible_by_4_unitary}
        U_{\text{even}} = \text{SWAP}_{(n/2), n-1} \text{SWAP}_{[(n/2)-2], n-3} \cdots \text{SWAP}_{4, [(n/2)+3]} \text{SWAP}_{2, [(n/2)+1]}.
    \end{align}
    For an odd number of ebits, we first perform 
    \begin{align}
        U_1=\text{SWAP}_{[(n/2)+1],[(n/2)+2,(n/2)+3,\cdots,n]}.
    \end{align} 
    This leads to 
    \begin{align}
        \op{i_1,i_1,i_2, i_2,\cdots, i_{\lfloor k/2 \rfloor}, i_{\lfloor k/2 \rfloor}, i_{\lfloor k/2 \rfloor +1}}{i_{\lfloor k/2 \rfloor +2},i_{\lfloor k/2 \rfloor +2},\cdots, i_{k}, i_{k}, i_{\lfloor k/2 \rfloor +1}}.
    \end{align}
    Note that the qubit labels reset after $U_1$. Then perform
    \begin{align}
        U_2 = \text{SWAP}_{[(n/2)-1], n-2} \text{SWAP}_{[(n/2)-3], n-4} \cdots \text{SWAP}_{4, [(n/2)+3]} \text{SWAP}_{2, [(n/2)+1]}.
    \end{align}
    Thus, for an odd number of ebits
    \begin{align} 
        U_{\text{odd}}=U_2U_1.
    \end{align}

\section{Additional Examples}
\label{App:4}

\subsection{Example 3}\label{subsec:ex3}
If we allow the asymmetric depolarizer to depolarize more in some directions than the symmetric depolarizer, we can get even higher values for the $M_1$ measure in some cases.
Take for example the following valid NCP map given by
\begin{align}\label{eq:Ex:maxDepolInSomeDirs}
    A = 
    \left[
    \begin{array}{cccc}
       -1/\sqrt{2} & 0 & 0 & 1 - 1/\sqrt{2} \\
       -1/\sqrt{2} & -1 & 0 & -1/\sqrt{2} \\
       -1/\sqrt{2} & 0 & -1 & -1/\sqrt{2} \\
       1 + 1/\sqrt{2} & 0 & 0 & 1/\sqrt{2}
    \end{array}
    \right],
\end{align}
which has an associated B-matrix with eigenvalues $2,-\sqrt{2},\sqrt{2},$ and $0$. The optimal symmetric depolarizer that makes the composition CP is given by the parameter 
\begin{align}
    \kappa = 1/(1+2 \sqrt{2}).    
\end{align}
It turns out that for any $\alpha, \beta, \gamma \geq \kappa$ for the asymmetric depolarizer, and not all equal to $\kappa$, we cannot make the composition CP. An optimization shows that we can completely depolarize in the x and z direction, while leaving the y direction unchanged; that is, 
\begin{align}
    \alpha = 0 = \gamma \text{ and }\beta = 1.
\end{align}
The $M_1$ measure is optimal and higher for the asymmetric depolarizer, even though we depolarized more than the symmetric depolarizer in two of the directions. To maintain complete positivity of the composition map while the magnitude of the depolarizing is less than or equal to the symmetric depolarizer in each direction, we must include negative values for one or two of the asymmetric depolarizing parameters. Depending on what measure you choose to use, $M_1$ measure or fidelity, you may want to depolarize more in some directions and less in others over the symmetric depolarizer.

In this next example, our $M_1$ measure can only be optimized when we depolarize more than the symmetric depolarizer in two of the directions. Define the NCP map with the evolutionary operator
\begin{align}
    A = 
    \left[
    \begin{array}{cccc}
       1-1/\sqrt{2} & 0 & 0 & -1/\sqrt{2} \\
       -1/\sqrt{2} & 1 & 0 & -1/\sqrt{2} \\
       -1/\sqrt{2} & 0 & 1 & -1/\sqrt{2} \\
       1/\sqrt{2} & 0 & 0 & 1+1/\sqrt{2}
    \end{array}
    \right]
\end{align}
which has an associated B-matrix with the eigenvalues $\{ 1 \pm \sqrt{2}, \pm 1 \}$. This map will translate the original Bloch sphere in the negative $x$ and $z$ directions by $-\sqrt{2}$ so that the original and image spheres touch at one point, i.e., there is only a single state in the domain defined by the state 
\begin{align}
    \rho = (1/2)[\mathbb{I} + (1/\sqrt{2})(\hat{\sigma}_x + \hat{\sigma}_z)]    
\end{align}
on the domain and $\hat{\sigma}_y \rho^T \hat{\sigma}_y$ on the image. If we symmetrically depolarize with parameter 
\begin{align}
    \kappa = 1/3,
\end{align} 
then the SPA becomes completely positive. The $M_1$ measure is equal to $1/3$ in this case, but we can obtain a value of $(1/3)*1.309$ when we use an asymmetric depolarizer with parameters 
\begin{align}
    \alpha = \gamma = 49/200 < \kappa \text{ and }\beta = 819/1000 > \kappa.
\end{align}

It turns out that we can only achieve an $M_1$ measure of up to $1/3$ when $|\alpha|, |\beta|, |\gamma| \geq 1/3$. You can see this by slightly bumping the absolute value of each parameter consecutively by a small value above $1/3$ and noticing that the composition map becomes NCP, as shown in the Mathematica file we uploaded to GitHub \cite{github_dilley}. Therefore, it is necessary to depolarize more than the symmetric depolarizer in two of the directions in order to achieve an advantage for the ADM. We also achieve a higher fidelity of Eq. (\ref{eq:fidelityErrorCorr}) in this case ($151/400 = 0.3775$) than we do in the symmetric one ($1/3$).

\subsection{Example 4}\label{subsec:ex4Repolar}
Let
\begin{align}
\mathcal{A}_\text{NCP}=\begin{bmatrix}
0 &0 &0 &1\\
0 &0 &x &0\\
0 &x &0 &0\\
1 &0 &0 &0
\end{bmatrix},
\qquad
\mathcal{B}_\text{NCP}=\begin{bmatrix}
0 &0 &0 &0\\
0 &1 &x &0\\
0 &x &1 &0\\
0 &0 &0 &0
\end{bmatrix},
\end{align}
where $x>1$. Then, the composition $\mathcal{L}(1/x,-1/x,-1)\circ \mathcal{A}_\text{NCP}$ is CP, where
\begin{align}
\mathcal{L}(1/x,-1/x,-1)=\begin{bmatrix}
0 &0 &0 &1\\
0 &0 &1/x &0\\
0 &1/x &0 &0\\
1 &0 &0 &0
\end{bmatrix},
\qquad
\mathcal{B}_{\mathcal{L}}=\begin{bmatrix}
0 &0 &0 &0\\
0 &1 &1/x &0\\
0 &1/x &1 &0\\
0 &0 &0 &0
\end{bmatrix}.
\end{align}
Notice how the initial and final states are left invariant after this transformation due to the fact that $\mathcal{L} \circ \mathcal{A} = \mathbb{I} \otimes \mathbb{I}$. This will always be true if $\mathcal{A}$ is invertible and its inversion has a positive dynamical B-matrix. This leads to the open question of which full-rank invertible A-matrices have non-negative dynamical B-matrices for the inverse? In the case of asymmetric depolarization, the initial parameters $\alpha, \beta, \gamma$ must be greater than 1 or less than $-1$. The translation NCP map will always have a negative inversion. There exists other $\mathcal{A}$ maps that cannot be expressed in terms of translations and asymmetric depolarization. Generally, it is unknown for these maps which will have CP inversions.

\subsection{Example 5}
Let
\begin{align}
\mathcal{A}_\text{NCP}=\begin{bmatrix}
1/2 &0 &0 &3/2\\
0 &x &0 &0\\
0 &0 &x &0\\
1/2 &0 &0 &-1/2
\end{bmatrix},
\qquad
\mathcal{B}_\text{NCP}=\begin{bmatrix}
1/2 &0 &0 &x\\
0 &3/2 &0 &0\\
0 &0 &1/2 &0\\
x &0 &0 &-1/2
\end{bmatrix},
\end{align}
where $x\not\in (-2,2)$. Then, the composition $\mathcal{L}(1/x^2,1/x^2, 0)\circ \mathcal{A}_\text{NCP}$ is CP, where
\begin{align}
\mathcal{L}(1/x^2,1/x^2, 0)=\begin{bmatrix}
1/2 &0 &0 &1/2\\
0 &1/x^2 &0 &0\\
0 &0 &1/x^2 &0\\
1/2 &0 &0 &1/2
\end{bmatrix},
\qquad
\mathcal{B}_\mathcal{L}=\begin{bmatrix}
1/2 &0 &0 &1/x^2\\
0 &1/2 &0 &0\\
0 &0 &1/2 &0\\
1/x^2 &0 &0 &1/2
\end{bmatrix}.
\end{align}

\section*{References}
\bibliographystyle{unsrt}

\end{document}